\documentclass[a4paper]{article}
\usepackage{graphicx} 
\usepackage{subcaption}
\usepackage{color}
\usepackage{amsfonts}
\usepackage{amsmath}
\usepackage{amssymb}
\usepackage{dsfont}
\usepackage{mathtools}
\usepackage{physics}
\usepackage{dcolumn}
\usepackage{bm}
\usepackage{cancel}
\usepackage{tensor}
\usepackage{amsthm}
\usepackage{sectsty}
\usepackage[toc,page]{appendix}
\usepackage{wrapfig}

\newtheorem{lemma}{Lemma}
\newtheorem{proposition}{Proposition}
\newtheorem{remark}{Remark}
\newtheorem{theorem}{Theorem}

\newtheorem*{fact}{Fact}

\newcommand*{\Z}{\makebox[1ex]{\textbf{$\cdot$}}}%

\allowdisplaybreaks

\newcounter{mnotecount}

\newcommand{\mnotex}[1]
{\protect{\stepcounter{mnotecount}}$^{\mbox{\footnotesize $\bullet$\themnotecount}}$ 
\marginpar{
\raggedright\tiny\em
$\!\!\!\!\!\!\,\bullet$\themnotecount: #1} }

\def\bma{{\bm a}}
\def\bmb{{\bm b}}
\def\bmc{{\bm c}}
\def\bmd{{\bm d}}
\def\bme{{\bm e}}
\def\bmf{{\bm f}}
\def\bmg{{\bm g}}
\def\bmh{{\bm h}}
\def\bmi{{\bm i}}
\def\bmj{{\bm j}}
\def\bmk{{\bm k}}
\def\bml{{\bm l}}
\def\bmn{{\bm n}}
\def\bmm{{\bm m}}

\def\bmp{{\bm p}}
\def\bmq{{\bm q}}
\def\bms{{\bm s}}

\def\bmu{{\bm u}}
\def\bmv{{\bm v}}

\def\bmomega{{\bm \omega}}

\def\bmL{{\bm L}}

\def\bmQ{{\bm Q}}

\def\bmT{{\bm T}}

\graphicspath{{../Figures/}}

\begin{document}

\title{\textbf{Future stability of self-gravitating  dust balls in an expanding universe}}

\author{S Beheshti\footnote{\tt s.beheshti@qmul.ac.uk}, M Normann\footnote{\tt m.normann@qmul.ac.uk} and  JA Valiente Kroon\footnote{\tt j.a.valiente-kroon@qmul.ac.uk}\\
{\em School of Mathematical Sciences}\\
{\em Queen Mary, University of London} \\
{\em Mile End Road, London E1 4NS, UK}}


\date{\today}

\maketitle

\begin{abstract}
 We consider a system representing self-gravitating balls of dust in an expanding Universe. It is demonstrated that one can prescribe data for such a system at infinity and evolve it backward in time without the development of shocks or singularities. The resulting solution to the Einstein-$\lambda$-dust equations exists for an infinite amount of time in the asymptotic region of the spacetime. Furthermore, we find that if the density is small compared to the Cosmological constant, then it is possible to construct Cosmological solutions to the Einstein constraint equations on a standard Cauchy hypersurface representing self-gravitating balls of dust. If, in addition, the density is assumed to be sufficiently small, then this initial data gives rise to a future geodesically complete solution to the 
 Einstein-$\lambda$-dust equations admitting a smooth conformal extension at infinity  which can be regarded as a perturbation of de Sitter spacetime. The main technical tool in this analysis are Friedrich's conformal Einstein field equations for the Einstein-$\lambda$-dust written in terms of a gauge in which the flow lines of the dust are recast as conformal geodesics.
\end{abstract}

\section{Introduction}
Much of Physics is the study of evolution of a system under certain conditions and laws. In Cosmology one is thus interested in the evolution of our Universe from the far past to the distant future. The dominant law governing galaxies and the evolution of the Cosmos is embedded in Einsteins theory of gravitation. In the large scale structure of the Universe, galaxies can be treated as dust ---i.e. each galaxy is represented by a "dust" particle--- exerting no pressure on the surrounding particles. Given that current observations suggest our Universe is expanding, the setting to investigate the evolution of our Universe is thus the Einstein equations coupled to dust matter with a positive cosmological constant. There is a challenge associated to the study of our Universe at such large scales and its evolution over long time: namely, the global properties of the theory become important. Accordingly, any comprehensive study of a solution to the Einstein field equations should also take into account its global properties.

\medskip
In the attempt to model astrophysical objects such as stars and galaxies and solar systems, it is necessary to consider solutions of the Einstein field equations which represent an \textit{isolated system}. This is far from a trivial endeavour. In Newtonian gravity and relativistic electrodynamics, one has a flat background metric upon which the fields propagate, and one can meaningfully speak about the fall-off properties of the fields as one moves away from the sources. In these terms an isolated system is a system for which the field strength vanishes at infinity and the source density is zero outside a finite radius. In General Relativity, the metric is part of the unknowns for which one solves the equations. Thus, there is no "background" metric upon which the gravitational field propagate and in terms of which we may define fall-off properties in a meaningful way. Accordingly,  attempts at solving the Einstein equations for an isolated system by introducing approximations in terms of a background metric plus perturbations cannot be satisfactory as they disregard the non-linear aspect of the full theory. 

\smallskip
A procedure which has proved successful in the study of global properties of spacetimes describing isolated systems   was devised by \textit{Roger Penrose} in \cite{Penrose1965}, and involves a conformal compactification of space time ---essentially allowing for a treatment of infinity as a three dimensional submanifold; see \cite{Kroon2016} for details. This allow for a rigorous description of the asymptotic behaviour and global properties of a space time \cite{Friedrich2014,Frauendiener2004,Newman1981}. But not all spacetimes allow for a conformal treatment. Accordingly,  one is interested in knowing which solutions to the Einstein field equations admit a smooth conformal compactification.

\smallskip
One important aspect of Penrose's conformal method which will be extensively used in the present article, is that a small conformal time can represent an infinite amount of physical time. Hence, if the equations describing the conformally rescaled spacetime imply a regular system of evolution equations one could, in turn,  apply general results of the theory of partial differential equations to show global existence and stability. The seminal work of H. Friedrich has established that the Einstein vacuum equations \cite{H.Friedrich1981}, including de Sitter-like spacetimes with positive cosmological constant \cite{friedrich1986a,Friedrich1986b}, the Einstein-Maxwell-Yang-Mills equations \cite{Friedrich1991} and the Einstein-$\lambda$-dust equations \cite{Friedrich2017}, all can be described in terms of a set of regular \emph{conformal Einstein field equations} 
from which, in turn, one can extract a symmetric hyperbolic evolution system for which general theory is available ---in particular, locally the Cauchy problem is well posed and stability over a small time is guaranteed.

\smallskip
 In \cite{Choquet-Bruhat2006}, Y. Choquet-Bruhat \& H. Friedrich have established the local existence in time of solutions to the Einstein field equations representing isolated self-gravitating dust bodies. However, the mathematical technology available did not allow to pursue the pressing question of the global existence of solutions. One of the key technical aspects of their analysis is the use of a formulation of the evolution equations which is well behaved independently of whether the density of the dust vanishes or not. This formulation crucially depends on the fact that the flow lines of the dust are geodesics. 
 
 A suitable framework for the analysis of global properties of solutions to the Einstein-$\lambda$-dust system by means of conformal methods was given in \cite{Friedrich2017}. This setup was used to study the backwards evolution of asymptotic data prescribed on the conformal boundary $\mathcal{I}^+$. The work in  \cite{Friedrich2017} is remarkable in that it is one of the few conformal treatments of a matter model with non-vanishing trace of the energy-momentum. The conformal evolution system  used in this analysis is well-defined up to and beyond the conformal boundary. Its construction  depends crucially on the observation that the flow lines of the dust can be recast as certain conformally invariant curves ---the so-called \emph{conformal geodesics}. Moreover, as in the case of the analysis in \cite{Choquet-Bruhat2006}, the evolution system is also regular independently of whether the density vanishes or not. Accordingly,  as it will be discussed in this article, it provides an ideal framework to study global properties of the evolution of isolated dust bodies in General Relativity in the presence of a positive Cosmological constant. The analysis of these relativistic self-gravitating matter configurations is a subject of physical relevance as the Cosmological constant is generally believed to be connected with the observed expansion of our Universe, and dust to the  solutions to the Einstein field equations are good models for the description of the matter content of the Universe.

\medskip
In this work we combine the approaches followed in \cite{Choquet-Bruhat2006} and \cite{Friedrich2017} to provide a toy model of self-gravitating dust balls in an expanding Universe for which it is possible to make assertions regarding global existence and stability. More precisely, we show that in a spacelike conformal boundary (which, for simplicity one can assume as having the topology of $\mathbb{S}^3$) one can prescribe asymptotic data which represents patches of dust on the conformal boundary. Using then the conformal evolution equations one can then show that these configurations will exist for some small amount of conformal time ---which, when translated into the  physical picture corrresponds to an infinite amount of physical time. To complement the above \emph{backwards evolution problem}, we provide sufficient conditions for the existence of solutions to the Einstein constraint equations on a standard Cauchy initial hypersurface which represent patches of dust in a de-Sitter-like universe. We further show future stability of these patches, provided the density function satisfies a smallness condition. The resulting spacetime is future geodesically complete. \emph{The above analysis provides a non-trivial example of fairly generic matter configurations which exist arbitrarily into the future.} The physical mechanism ensuring this result is the  expansion driven by the Cosmological constant $\lambda$. 


\subsection*{Outline of the article}
The structure of the article is as follows. In Section \ref{sectionGeometricBackground} we introduce the geometric framework we will be working in. Section \ref{secFriedrichReview} gives a brief discussion of the conformal evolution system for the Einstein-$\lambda$-dust system of equations that will be used in our analysis. This framework is based on the work in \cite{Friedrich2017}  ---this discussion includes a brief discussion on the conformal method.  Section \ref{secPastStability} discuss the construction of asymptotic data for self-gravitating dust balls and provides a proof of the evolution, backwards in time, of this data. In Section \ref{sec:futureStability} we show how the conformal method of Licnerowicz can be used to obtain solutions to the constraints describing an initial configuration of dust balls. Moreover, it is proven that under certain smallness assumptions this initial data gives rise to a future geodesically complete spacetime with positive Cosmological constant containing self-gravitating dust balls. The final Section \ref{Section:FinalRemarks} provides some concluding remarks.

\subsection*{Conventions and notation}
In this article, Lorentzian metrics are assumed to have signature $(-+++)$. Throughout, a coordinate free notation will be preferred where it is most convenient. To avoid confusion, we specify here the precise meaning of some symbols. Where indices are used, the following hold: Greek and Latin letters will be used as coordinate
indices in the spacetime manifold, where $\mu,\,\nu,\,\lambda,\ldots =
\{0,1,2,3\}$ and $i,j,k \ldots = \{1,2,3\}$. To denote frame indices
we will make use of bold latin letters where $\bma, \bmb, \bmc \ldots
= \{0,1,2,3\}$ and $\bmi, \bmj, \bmk \ldots = \{1,2,3\}$. Hence, the
components in a frame basis of a vector $\bmv \in \mathcal{M}$ is thus
labelled $v^{\bma}$.
\\

\noindent \textbf{Derivative expressions.} Given local coordinates $x^{\mu}$ and a covector $\bm\omega$, we write
\begin{eqnarray*}
&&\bm \nabla \bm \omega \equiv \nabla_{\mu} \omega_{\nu} \mathbf{d}x^{\mu}\otimes  \mathbf{d}x^{\nu},\\
&&\bm \nabla \bm\cdot \bm\omega \equiv \nabla^{\mu} \omega_{\mu},\\
&&\bm\nabla^2\bm\omega \equiv g^{\mu\nu}\nabla_{\mu}\nabla_{\nu}\omega_{\gamma} \mathbf{d}x^{\gamma},\\
&&\bm \nabla \bm \nabla \bm\omega \equiv \nabla_{\mu}\nabla_{\nu}\omega_{\gamma} \mathbf{d}x^{\mu}\otimes \mathbf{d}x^{\nu} \otimes \mathbf{d}x^{\gamma}.
\end{eqnarray*}
where $\nabla_{\mu}$ denotes the usual covariant derivative, and $g_{\mu\nu}$ the components of the metric tensor expressed in a local coordinate system.
\\

\noindent \textbf{Covariant and contravariant tensors}. We will use the musical notation $\bm\sharp$ and $\bm\flat$ to distinguish between contravariant and covariant tensors, respectively. Given the coordinates $x^{\mu}$, we may write
\[
\bmg = g_{\mu\nu}\mathbf{d}x_{\mu} \otimes \mathbf{d}x_{\nu}, \qquad 
\bmg^{\bm\sharp} = g^{\mu\nu}\bm\partial_{\mu} \otimes\bm\partial_{\nu}.
\]
Note, however, that this will not be used consistently, but only where it is important to distinguish between the two. The introduction of indexfree notation allow the discussion on the conformal field equations to focus on the structural aspects, and thus be more conseptual, rather then paying attention to the gory details.

\section{Geometric background}
\label{sectionGeometricBackground}
In what follows, let $(\mathcal{M},\bmg)$ denote a spacetime represented by a 4-dimensional
manifold, $\mathcal{M}$, with a Lorentzian metric $\bmg$. In the following we assume the existence of a smooth timelike vector field $\bmu$. Moreover, it is assumed that 
\[
\bmg(\bmu, \bmu) = -1,
\]
and that at each point $p \in \mathcal{M}$ there exists a $\bmg$-orthonormal frame field ---i.e. it is assumed that 
\[
\bmg(\bme_{\bma},\bme_{\bmb}) = \eta_{\bma\bmb}.
\]
Associated to the frame $\{\bme_{\bma}\}$ one has a coframe, $\{{\bmomega}^{\bma}\}$ satisfying
\[
\langle{\bme_{\bma}, \mathbf{\bmomega}^{\bmb}\rangle} = \delta_{\bma}{}^{\bmb}.
\]
In the following all tensorial objects will be expressed in terms of their components with respect to the frame and coframe unless otherwise stated. The metric tensor $\bmg$ gives rise to a natural connection $\mathbf{\nabla}$ such that $\mathbf{\nabla} \bmg = 0$,
which is the \textit{metric compatibility condition}. In terms of the frames, this condition takes the form
\begin{equation}
\label{metricComp}
\Gamma_{\bma}{}^{\bmb}{}_{\bmc} \eta_{\bmb\bmd} + \Gamma_{\bma}{}^{\bmb}{}_{\bmd} \eta_{\bmb\bmc} = 0,
\end{equation}
where the \textit{frame connection coefficients} are defined by the directional derivative along the direction of the frame indices
\[
\label{FrameConnectionDef}
\nabla_{\bma} \bme_{\bmb} = \Gamma_{\bma}{}^{\bmc}{}_{\bmb} \bme_{\bmc}, \qquad \nabla_{\bma} = \langle{\bme_{\bma}, \mathbf{\nabla}\rangle}.
\]
Furthermore, if the connection $\mathbf{\nabla}$ is \textit{torsion-free}, we have that
\begin{equation}
\label{torsionFree}
\Sigma_{\bma}{}^{\bmc}{}_{\bmb} = 0,
\end{equation}
where the frame components of the \textit{torsion tensor} are defined by
\[
\Sigma_{\bma}{}^{\bmc}{}_{\bmb} \bme_{\bmc} = \left[\bme_{\bma}, \bme_{\bmb}\right] + \left(\Gamma_{\bma}{}^{\bmc}{}_{\bmb} - \Gamma_{\bmb}{}^{\bmc}{}_{\bma}\right) \bme_{\bmc}.
\]
The commutation of the connection may be expressed in terms of the \textit{Riemann curvature tensor} and the torsion tensor
\begin{eqnarray*}
&&\nabla_{[\bma}\nabla_{\bmb]}v^{\bmc} = R^{\bmc}{}_{\bmd\bma\bmb} v^{\bmd} +\Sigma_{\bma}{}^{\bmd}{}_{\bmb}\nabla_{\bmd}v^{\bmc},\\
&&\nabla_{[\bma}\nabla_{\bmb]}w_{\bmc} = -R^{\bmd}{}_{\bmc\bma\bmb} w_{\bmd} +\Sigma_{\bma}{}^{\bmd}{}_{\bmb}\nabla_{\bmd} w_{\bmc}\label{CommutatorRiemann}.
\end{eqnarray*}
The frame components of the Riemann curvature tensor is given by
\begin{equation}
R^\bmc{}_{\bmd\bma\bmb} = \partial_\bma\Gamma_\bmb{}^\bmc{}_\bmd
   - \partial_\bmb \Gamma_\bma{}^\bmc{}_\bmd +
   \Gamma_\bmf{}^\bmc{}_\bmd(\Gamma_\bmb{}^\bmf{}_\bma -
   \Gamma_\bma{}^\bmf{}_\bmb) +
   \Gamma_\bmb{}^\bmf{}_\bmd\Gamma_\bma{}^\bmc{}_\bmf -
   \Gamma_\bma{}^\bmf{}_\bmd\Gamma_\bmb{}^\bmc{}_\bmf
   -\Sigma_\bma{}^\bmf{}_\bmb \Gamma_\bmf{}^\bmc{}_\bmd \label{RiemannExpansion}
\end{equation}
---see \cite{Kroon2016} for details. The Riemann tensor has all the usual symmetries, and it satisfies the
\textit{Bianchi identity} for a general connection
\begin{eqnarray}
&& R^{\bmd}{}_{[\bmc\bma\bmb]} + \nabla_{[\bma}\Sigma_{\bmb}{}^{\bmd}{}_{\bmc]} +
   \Sigma_{[\bma}{}^{\bme}{}_{\bmb}\Sigma_{\bmc]}{}^{\bmd}{}_{\bme}=0, \label{1stBianchiId}\\
&& \nabla_{[\bma} R^{\bmd}{}_{|\bme|\bmb\bmc]} + \Sigma_{[\bma}{}^{\bmf}{}_{\bmb} R^{\bmd}{}_{|\bme\bm|\bmc]\bmf} =0.\label{2ndBianchiId}
\end{eqnarray}
Furthermore, we recall that the Riemann tensor admits the \emph{irreducible decomposition}
\begin{eqnarray}
&& R^{\bmc}{}_{\bmd\bma\bmb} = C^{\bmc}{}_{\bmd\bma\bmb} + 2 (\delta^{\bmc}{}_{[\bma}L_{\bmb]\bmd} -
   \eta_{\bmd[\bma}L_{\bmb]}{}^{\bmc}), \label{RiemannDecomposition}
\end{eqnarray}
with $C^{\bmc}{}_{\bmd\bma\bmb}$ the components of the \emph{Weyl tensor} and 
\begin{equation}
L_{\bma\bmb} \equiv R_{\bma\bmb} -\frac{1}{6}R \eta_{\bma\bmb}
\label{Definition:Schouten}
\end{equation}
denotes the components of the \emph{Schouten tensor}. The connection
$\mathbf{\nabla}$ is called the \textit{Levi-Civita connection} of $\bmg$ if it
satisfies \eqref{metricComp} and \eqref{torsionFree}. In what follows we
will assume the connection to be Levi-Civita.
\\

\subsubsection*{A projection formalism}
At each point in the spacetime manifold $\mathcal{M}$ the flow lines give rise to a tangent
space which can be split into parts in the direction of $\bmu$
and those orthogonal. This means that without implying a foliation, we
may decompose every tensor defined at each point $p \in \mathcal{M}$
into its orthogonal and timelike part. This may be done by contracting
with $\mathbf{u}$ and the \textit{projector} defined as
\[
h_{\bma}{}^{\bmb} \equiv \eta_{\bma}{}^{\bmb} + u_{\bma}u^{\bmb}, \qquad \bmu = u^{\bma}\mathbf{e}_{\bma}.
\]
Thus, a tensor $T_{\bma\bmb}$ may be split into its time-like, mixed
and space-like parts given, respectively, by
\[
T_{\bm0\bm0}= u^{\bma}u^{\bmb}T_{\bma\bmb}, \qquad T'_{\bm0\bmc}= u^{\bma}h^{\bmb}{}_{\bmc}T_{\bma\bmb}, \qquad T'_{\bmc\bmd}= h^{\bma}{}_{\bmc}h^{\bmb}{}_{\bmd}T_{\bma\bmb},
\]
where $'$ denotes that the free indices left are spatial ---e.g. $T'_{\bma\bm0} u^{\bma} = 0$. In what follows we will not put a $'$ on projected tensors, as we will mostly deal with timelike parts and pure spatial parts --- in which case we use the components $\{\bmi,\bmj,...\}$. Decomposing $\mathbf{\nabla u}$ we
obtain
\begin{equation}
\label{Der4VelDecomp}
\nabla_{\bma} u^{\bmb} = \chi_{\bma}{}^{\bmb} + u_{\bma}a^{\bmb},
\end{equation}
where $\chi_{\bma}{}^{\bmb}$ and $a^{\bmb}$ are the components of the
\textit{Weingarten tensor} and 4-acceleration, respectively, defined
by
\[
\chi_{\bma}{}^{\bmb} \equiv h_{\bma}{}^{\bmc}\nabla_{\bmc} u^{\bmb}, \qquad a^{\bmb} \equiv u^{\bmc}\nabla_{\bmc}u^{\bmb}.
\]
In the literature (e.g. see \cite{Wald1984} p.217) the trace,
trace-free and antisymmetric part of \eqref{Der4VelDecomp} is called,
respectively, the expansion, shear and the twist of the fluid.
Given a tensor $T_{abc}$ which is antisymmetric in its two last indices,
we may construct the \textit{electric} and \textit{magnetic} parts with respect to $\mathbf{u}$. In frame indices this is, respectively, defined by
\[
\omega_{\bmc\bmd} \equiv T_{\bma\bmb\bme} h_{\bmc}{}^{\bma} h_{\bmd}{}^{\bmb} u^{\bme}, \qquad \omega^{\ast}_{\bmc\bmd} \equiv T^{\ast}{}_{\bma\bmb\bme} h_{\bmc}{}^{\bma} h_{\bmd}{}^{\bmb} u^{\bme},
\]
where the \textit{Hodge dual operator}, denoted by ${}^{\ast}$, is defined
by 
\[
T^{\ast}{}_{\bma\bmb\bme} \equiv -\frac{1}{2}\epsilon^{\bmm\bmn}{}_{\bmb\bme} T_{\bma\bmm\bmn},
\]
and has the property that
\[
T^{\ast \ast}{}_{\bma\bmb\bmc} = -T_{\bma\bmb\bmc}.
\]
Depending on the symmetries and rank of the tensor, the above
definition for electric and magnetic decomposition may vary
slightly. Central for our discussion is that $omega_{\bma\bmb}$ and $\omega^{\ast}_{\bma\bmb}$
are spatial and symmetric.

\section{Conformal Einstein dust flows with positive cosmological constant}
\label{secFriedrichReview}
In the following let $(\mathcal{\tilde{M}}, \tilde{\bmg})$ denote a spacetime satisfying the Einstein field equations. Following the standard usage we call $(\mathcal{\tilde{M}}, \tilde{\bmg})$ the \emph{physical spacetime}. We will also consider a conformally rescaled spacetime $\left(\mathcal{M}, \bmg\right)$, the \emph{unphysical spacetime}, whose metric satisfies
\begin{equation}
\label{ConformalMetricDef}
    \bmg = \Omega^2 \tilde{\bmg}, \qquad \bmg^{\bm\sharp} = \Omega^{-2} \tilde{\bmg}^{\bm\sharp},
\end{equation}
 with $\Omega$ a smooth function defined on $\mathcal{M}$ which plays the role of a \emph{boundary defining function}. The conformal boundary $\mathcal{I}^+$ is defined as all the points where $\Omega$ vanish --- i.e. $\mathcal{I}^+ \equiv \{p \in \mathcal{M} \ | \ \Omega(p) = 0\}$. Again, following standard usage we identify the interior of $\mathcal{M}$ by $\mathcal{\tilde{M}}$, and adopt the standard convention that fields with/without a tilde indicate they are defined in relation to the physical/unphysical spacetime.


\subsection{Einstein-$\lambda$-dust system}
In this article we are concerned with the Einstein-$\lambda$-dust system governed by the equations 
\begin{subequations}
\begin{eqnarray}
    &&\mathbf{Ric}[\tilde{\bm g}] - \left(\frac{1}{2}R[\tilde{\bm g}] - \lambda \right)\tilde{\bmg} = \kappa \tilde{\bm T},\label{PhysicalEFE}\\
    &&\tilde{\bm T} = \tilde{\rho}\tilde{\bm U} \otimes \tilde{\bm U},\label{PhysicalEnergyMomentum}\\
    &&\tilde{\bm\nabla}_{\bm \tilde{U}}\tilde{\bm U} = 0,\label{PhysicalMatterEq1}\\
    &&\tilde{\bm \nabla} \bm\cdot \tilde{\bm j} = 0,\label{PhysicalMatterEq2}
\end{eqnarray}
\end{subequations}
where $\mathbf{Ric}[\tilde{\bm g}]$, $\tilde{\bm T}$ and $\tilde{\bm U} $ are the Ricci tensor, energy momentum tensor for the physical metric and the four velocity of the particle trajectories, respectively. Furthermore, we have defined the matter current $\tilde{\bm j}$ as,
\[
\tilde{\bm j} \equiv \tilde{\rho}\tilde{\bm U}, 
\]
where $\tilde{\rho}$ is a positive function representing the energy-density of the matter. We also let $\kappa$ and $\lambda$ be positive constants. In the following we shall set $\kappa = 1$ to simplify the discussion.

\begin{remark}
\em Note that equations \eqref{PhysicalMatterEq1} and \eqref{PhysicalMatterEq2} are the equations of motion for the matter fields obtained through the divergence-free condition $\tilde{\bm\nabla} \bm\cdot \tilde{\bmT} = 0$. In particular, equation \eqref{PhysicalMatterEq1} states that the flow lines of the dust matter model are geodesics. 
\end{remark}

\subsection{Conformal transformations}
 The conformal transformation \eqref{ConformalMetricDef} implies the following relationship between the two respective connections,
\begin{equation}
\left(\bm\nabla- \tilde{\bm\nabla}\right)\bm\omega = \bmQ\cdot\bmomega\label{coformalConnectionTrans},
\end{equation}
where $\bm Q$ is a symmetric 3-rank tensor defined in local coordinates $x^{\mu}$ by
\[
\bm Q \equiv \Omega^{-1}\left(\nabla_{\delta}\Omega g_{\mu\nu}g^{\delta\gamma} - \nabla_{\mu}\Omega \delta^{\gamma}{}_{\nu} - \nabla_{\nu}\Omega \delta^{\gamma}{}_{\mu}\right)\mathbf{d}x^{\mu}\otimes \bm\partial_{\gamma} \otimes \mathbf{d}x^{\nu}.
\]
Furthermore we introduce an \emph{unphysical 4-velocity} $\bm U$ related to the physical tangent vector to the flow lines $\tilde{\bm U}$ in such a way that $\bm g(\bm U,\bm U) = \tilde{\bm g}(\tilde{\bm U},\tilde{\bm U}) = -1$ --- i.e. we have

\[
\bm U = \Omega^{-1} \tilde{\bm U} \qquad \bm U^{\bm\flat} = \Omega \tilde{\bm U}^{\bm\flat}.
\]


\smallskip
\noindent Since $\tilde{\rho}$ is a scalar field independent of the metric, its transformation rule can be freely specified. It is convenient to define the \emph{unphysical energy density} as 
\[
\rho = \Omega^{-3}\tilde{\rho}.
\]

\smallskip
\noindent Using the conformal transformation \eqref{ConformalMetricDef} one can thus use \eqref{coformalConnectionTrans} to obtain the transformation rule of Ricci tensor of the physical metric $\mathbf{Ric}[\tilde{\bmg}]$, 
\begin{equation}
    \mathbf{Ric}[\bm g] - \mathbf{Ric}[\tilde{\bmg}] = - 2\Omega^{-1}\bm\nabla \bm\nabla\Omega - \bmg \otimes \left(\Omega^{-1}\bm \nabla^2\Omega - 3\Omega^{-2}(\bm\nabla \Omega)^2 \right).\label{RicciTensorTrans}
\end{equation}
Contracting the above equation with  $\bm g$, we find the transformation of the Ricci scalar $R[\tilde{\bm g}]$,
\begin{equation}
 R[\bm g] - \Omega^{-2}R[\tilde{\bm g}] = -6\Omega^{-1}\bm\nabla^2\Omega + 12\Omega^{-3}\left(\bm\nabla\Omega\right)^2.\label{RicciScalarTrans}
\end{equation}
\begin{remark}
\em Observe that equations \eqref{RicciTensorTrans} and \eqref{RicciScalarTrans} are singular at the points where $\Omega = 0$. Thus, using the form of the Ricci tensor and scalar as above will not directly lead to a set of field equations which extends to the conformal boundary $\mathcal{I}^{+}$. It is, therefore, necessary to find another set of equations which are equivalent to \eqref{RicciTensorTrans} and \eqref{RicciScalarTrans} but that extends smoothly to the conformal boundary. Another issue is the freedom in the choice of the conformal factor $\Omega$. This gauge freedom means a solution to the equations \eqref{RicciTensorTrans} and \eqref{RicciScalarTrans} are, in some sense, not unique. Hence, the new set of equations must be constructed such that one can fix this gauge freedom. 
\end{remark}

\subsection{The conformal regular Einstein-$\lambda$-dust system}

 We refer to \cite{Kroon2016} for a derivation of the \textit{regular conformal field equations}. 
In what follows we will only give a short summary of some results from \cite{Friedrich2017} which make up the basis of the next section. We refer the interested reader to the original paper for details.

\medskip
\noindent As in Section \ref{sectionGeometricBackground}, we introduce a $\bmg$-orthonormal frame field $\{\bme_{\bma}\}$ on $\mathcal{M}$ such that in local coordinates $x=(x^{\mu})$ we have that $\bme_{\bma} = e^{\mu}{}_{\bma} \partial_{\mu}$. Furthermore, $g_{\bma\bmb}\equiv\bmg (\bme_{\bma}, \bme_{\bmb}) = \eta_{\bma\bmb}$ and we assume the frame connection defined by \eqref{FrameConnectionDef} to be such that the metric compatibility condition \eqref{metricComp} holds. Then, one can recover the physical metric $\tilde{\bm g}$ by the transformation \eqref{ConformalMetricDef}. Most equations and tensor fields are henceforth given in terms of frame indices.

\medskip
\noindent Using $\{\bme_{\bma}\}$ as the fundamental geometric unknown, we may write a new set of equations entirely in terms of fields on $\mathcal{M}$ which is equivalent to the system \eqref{PhysicalEFE}-\eqref{PhysicalMatterEq2} in the domain $\Omega > 0$ ---namely 
\begin{subequations}
\begin{eqnarray}
&& 6s \Omega - 3\nabla_{\bma}\Omega \nabla^{\bma}\Omega - \lambda = \frac{1}{4}\Omega^3 \rho,\label{ConformalEq1}\\
&&\nabla_{\bmb}\nabla_{\bmd}\Omega + \Omega L_{\bmb\bmd} - s g_{\bmb \bmd} = \frac{1}{2}\Omega^2 \rho \left(U_{\bmb}U_{\bmd} + \frac{1}{4}g_{\bmb\bmd}\right), \label{ConformalEq2}\\
&&\nabla_{\bmd}s + \nabla_{\bma}\Omega L^{\bma}{}_{\bmd} = \frac{1}{2}\nabla^{\bma}\Omega \rho \left(U_{\bma}U_{\bmd}+\frac{1}{4}g_{\bma\bmd}\right) + \frac{1}{8}\Omega\rho \nabla_{\bmd}\Omega + \frac{1}{24}\Omega^2\nabla_{\bmd}\rho, \label{ConformalEq3}\\
&&2\nabla_{[\bmd}L_{\bmc]\bmb} - \nabla_{\bma}\Omega W^{\bma}{}_{\bmb\bmd\bmc} = \Omega\left(\rho\left(\nabla_{[\bmd}U_{\bmc]}U_{\bmb} +U_{[\bmc}\nabla_{\bmd]}U_{\bmb}\right) + \nabla_{[\bmd}\rho U_{\bmc]}U_{\bmb} + \frac{1}{3} \nabla_{[\bmd}\rho g_{\bmc]\bmb}\right)\\
&& \qquad \qquad \qquad \qquad \qquad \qquad + \rho Z_{\bmb\bmd\bmc} \nonumber \label{ConformalEq4} \\
&&\nabla_{\bma}W^{\bma}{}_{\bmb\bmd\bmc} = \rho\left(\nabla_{[\bmd}U_{\bmc]}U_{\bmb} +U_{[\bmc}\nabla_{\bmd]}U_{\bmb}\right) + \nabla_{[\bmd}\rho U_{\bmc]}U_{\bmb} + \frac{1}{3} \nabla_{[\bmd}\rho g_{\bmc]\bmb} + \frac{\rho}{\Omega}Z_{\bmb\bmd\bmc} \label{ConformalEq5}.
\end{eqnarray}
\end{subequations}
The matter equations are given by,
\begin{subequations}
\begin{eqnarray}
&&U^{\bma}\nabla_{\bma}U^{\bmd} = \frac{1}{\Omega}\left(g^{\bmd}{}_{\bma} + U^{\bmd}U_{\bma}\right), \label{ConformalMatterEq1}\\
&&U^{\bma}\nabla_{\bma}\rho = -\rho\chi^{\bma}{}_{\bma} \label{ConformalMatterEq2}.
\end{eqnarray}
\end{subequations}
In the above, the following fields has been defined:
\begin{subequations}
\begin{eqnarray}
&&s \equiv \frac{1}{4} \nabla^{\bma}\nabla_{\bma}\Omega + \frac{1}{24}\Omega R[\bmg],\\
&&W^{\bmd}{}_{\bma\bmb\bmc} \equiv \Omega^{-1}C^{\bmd}{}_{\bma\bmb\bmc},\\
&&Z_{\bmb\bmd\bmc} \equiv \nabla_{[\bmd}\Omega g_{\bmc]\bmb} + 2\nabla_{[\bmd}\Omega U_{\bmc]}U_{\bmb} + U_{[\bmd}g_{\bmc]\bmb}g^{\bme\bmf}\nabla_{\bme}\Omega U_{\bmf}.
\end{eqnarray}
\end{subequations}
\begin{remark}
\em The main interest is to find solutions to the system \eqref{ConformalEq1}-\eqref{ConformalMatterEq2} in the domain $\Omega >0$ which admit a meaningful limit on $\mathcal{I}^+$. It was found by Friedrich that a necessary condition to have this type of solutions is that the geodesics generated by $\tilde{U}$ approach $\mathcal{I}^+$ orthogonally ---see \cite{Friedrich2017}.
\end{remark}
\noindent In the above and in what follows, the frame field is fixed by choosing $e_{\bm0} = \bm U$ and the \textit{Lagrangian gauge} --- i.e. given coordinates $x^{\mu}$ in a neighbourhood $\mathcal{U} \subset \mathcal{M}$ then the frame components of $e_{\bm 0}$ are given by $e_{\bm0}{}^{\mu} = \delta_{\bm0}{}^{\mu}$. Moreover, \textit{Fermi propagation} of the spatial components of the frame will be employed. More precisely, one has that,
\[
\Gamma_{\bm 0}{}^{\bma}{}_{\bmb} = 0.
\]

\subsection{Regularisation of the equations}
\label{Section:EvolutionEquations}

In order for the above system to be of use, it is necessary to deal with the singular equations \eqref{ConformalEq5} and \eqref{ConformalMatterEq1}. To do so, one makes use of the \textit{conformal geodesic equation},
\begin{subequations}
\begin{eqnarray}
&&\tilde{\nabla}_{\bm U}\bm U + 2\langle{\bm b, \bm V\rangle}\bm U - \tilde{\bmg}(\bm U,\bm U)\bm b^{\bm\sharp}=0\nonumber\\
&&\tilde{\nabla}_{\bm U}\bm b - \langle{\bm b, \bm U\rangle}\bm b + \frac{1}{2} \tilde{\bmg}(\bm b,\bm b)\bm U^{\bm\flat} - \bm \tilde{L}(\bm U, \Z )=0,\nonumber
\end{eqnarray}
\end{subequations}
where $\tilde{\bm{L}}$ is the Shouten tensor for the physical metric and $\bm b$ a one form associated with a curve $\gamma (\bm \tau)$ for which $\bm U$ is the tangent vector. A solution $(\bm b (\tau), \bm U(\tau) )$ to the conformal geodesic equation is called a \textit{conformal geodesic}.
\begin{remark}
\em The one form $\bm b$ can be thought of as an accelleration associated with $\bm U$. Thus, for $\bm b = 0$, a conformal geodesic coinsides with a metric geodesic.
\end{remark}
\noindent Given local coordinates $x=(x^{\mu})$ and a curve parametrised by $\sigma \in \mathds{R}$, Friedrich defines a new 1-form $\bmf$ with components,
\begin{equation}
f_{\nu}(\sigma) \equiv b_{\nu}(\sigma) - \Omega^{-1}\nabla_{\nu}\Omega \rvert_{x(\sigma)},
\end{equation}
where $b_{\nu}$ are the components of a one-form satisfying the geodesic equation. If, in addition, the components of a vector $\bm V$ are given by $V^{\mu}(\sigma)
 = \dfrac{dx^{\mu}}{d\sigma}$, then $\left( \bm f, \bm V\right)$ is a solution to the equations
 \begin{subequations}
\begin{eqnarray}
&&\nabla_{\bm V}\bm V + 2\langle{\bmf, \bm V\rangle}\bm V - \bmg(\bm V,\bm V)\bmf^{\bm\sharp}=0 \label{EvolEqForV}\\
&&\nabla_{\bm V}\bmf - \langle{\bmf, \bm V\rangle}\bmf + \frac{1}{2} \bmg(\bmf,\bmf)\bm V^{\bm\flat} - \bm L(\bm V, \Z )=0.\label{EvolEqForf}
\end{eqnarray}
\end{subequations}
\begin{remark}
\em Observe that equations \eqref{EvolEqForV} and \eqref{EvolEqForf} involves only conformal fields, as opposed to the conformal geodesic equation. 
\end{remark}
\noindent By assuming that $\bm V$ is related to the tangent vector of a geodesic of the matter particles via 
\[
\bm V = \omega^{-1} \tilde{\bm U}, \qquad \omega^{-1} \equiv \dfrac{dt}{d\sigma}
\]
and with the relations
\[
\bmg(\bm V, \bm V) = -\theta^{-2}, \qquad \theta = \frac{\omega}{\Omega}, \qquad \nabla_{\bm U}\theta = \theta \langle\bm U,\bm f \rangle,
\]
it can then be shown, using the definition for $\bmf$, that one obtains a \textit{regularising relation} which in frame indices takes the form,
\begin{equation}
\label{RegularisingEq}
    \nabla_{\bma}\Omega = - \left(\nabla_{\bm 0}\Omega +\Omega f_{\bm0}U_{\bm a}\right) - \Omega f_{\bm a}.
\end{equation}
Using the above equation in \eqref{ConformalEq5} and \eqref{ConformalMatterEq1} one removes the singularities in the system of equations \eqref{ConformalEq1}-\eqref{ConformalMatterEq2} which now can be smoothly extended to the conformal boundary. In other words, one has a regular system of field equations. Moreover, it can be shown that these equations imply a symmetric hyperbolic system of equations for the unknowns
\[
(e_{\bmi}{}^{\mu}, \Gamma_{\bmk}{}^{\bmi}{}_{\bmj}, f_{\bmd}, \varsigma_{\bmi\bmj},\xi,\Omega,\Sigma_{\bmd}, s, L_{\bm0\bmi}, L_{\bmi\bmj}, \rho, \omega_{\bmi\bmj}, \omega^{\ast}{}_{\bmi\bmj}).
\]
More precisely, one has the equations
\begin{subequations}
\begin{eqnarray}
&&\partial_{0} e_{\bmi}{}^{\mu} = -f_{\bmi}\delta_{\bm0}{}^{\mu} - \chi_{\bmi}{}^{\bmj}e_{\bmj}{}^{\mu}\label{ReducedEvolEq1},\\
&&\partial f_{\bm0} = -\frac{1}{2}f_{\bma}f^{\bma} + L_{\bm0\bm0},\\
&&\partial_{\bm0} f_{\bmi} = L_{\bm0\bmi},\\
&&\partial_{\bm0} \varsigma_{\bmi\bmj} = -\Omega\left(\varsigma_{\bmi}{}^{\bmk}\varsigma_{\bmk\bmj} - \frac{1}{3}\varsigma^{\bmk\bml}\varsigma_{\bml\bmk}g_{\bmi\bmj} \right) - \frac{2}{3}\left(\nabla_{\bm U}\Omega\right)^{-1}\left(\Omega \xi - 3s\right)\varsigma_{\bmi\bmj} - W_{\bm0 \bmi\bm0\bmj},\\
&&\partial_{\bm0}\xi = \left(\nabla_{\bm U}\Omega\right)^{-1}\left(\Omega \xi - 3s\right)\left(-\frac{1}{3}\xi + f_{\bmi}f^{\bmi} - L_{\bm0\bm0} + \frac{1}{4}\rho \Omega\right) - \nabla_{\bm U}\Omega \Omega\varsigma^{\bmk\bml}\varsigma_{\bml\bmk}\nonumber\\
&& \qquad + 3 f^{\bmi}L_{\bmi\bm0} - \frac{3}{4}\rho \nabla_{\bm U}\Omega,\\
&&\nabla_{\bm0}\Omega=\Sigma_{\bm0},\label{PropagationOmega}\\
&&\nabla_{\bm0}\Sigma_{\bmd} = -\Omega L_{\bm0\bmd} + sg_{\bm0\bmd} + \frac{1}{2}\Omega^2\rho \left(U_{\bm0}U_{\bmd} + \frac{1}{4}g_{\bm0\bmd}\right)\label{PropagationSigma},\\
&& \nabla_{\bm0}s = -\nabla^{\bma}\Omega L_{\bma\bm0}= \frac{1}{2}\Omega\rho\nabla^{\bma}\Omega\left(U_{\bm0}U_{\bmd} + \frac{1}{4}g_{\bm0\bmd}\right) + \frac{1}{8}\Omega\rho \nabla_{\bm0}\Omega + \frac{1}{24}\Omega^2\nabla_{\bm0}\rho,\\
&&\nabla_{\bm0}L_{\bm0\bmi} = h^{\bmi\bmj}\nabla_{\bmj}L_{\bmi\bmk} + \frac{1}{6}\nabla_{\bmi}R + K^{\bmb}{}_{\bmb\bmi},\\
&&\nabla_{\bm0}L_{\bmi\bmi}=\nabla_{\bmi}L_{\bm0\bmi} + K_{\bmi\bm0\bmi},\\
&&\nabla_{\bm0}L_{\bmi\bmj} = \nabla_{\bmi}L_{\bm0\bmj} + \nabla_{\bmj}L_{\bm0\bmi} + K_{\bmi\bm0\bmj} + K_{\bmj\bm0\bmi},\\
&&\nabla_{0}\omega_{\bmi\bmj} + D_{\bmk}\omega^{\ast}{}_{\bml(\bmj}\epsilon_{\bmi)}{}^{\bmk\bml} = L.O.T,\\
&&\nabla_{\bm0}\omega^{\ast}{}_{\bmi\bmj} - D_{\bmk}\omega_{\bml(\bmj}\epsilon_{\bmi)}{}^{\bmk\bml} = L.O.T,\label{ReducedEvolEqLast}
\end{eqnarray}
\end{subequations}
where $L.O.T$ stand for 'lower order terms.' In the above, the following fields have been defined,
\begin{subequations}
\begin{eqnarray}
&& \omega_{\bma\bmb} \equiv W_{\bmc\bmd\bm r\bms}U^{\bmc}U^{\bm r}h^{\bmd}{}_{\bma}h^{\bms}{}_{\bmb}, \qquad \omega^{\ast}_{\bma\bmb} \equiv  \frac{1}{2} W_{\bmc\bmd\bmp\bmq}\epsilon^{\bmm \bmn}{}_{\bm r\bms}U^{\bmc}U^{\bm r}h^{\bmd}{}_{\bma}h^{\bms}{}_{\bmb}\\
&& \varsigma_{\bmi\bmj} \equiv \Omega^{-1}\left(\xi_{\bmi\bmj} - \frac{1}{3}g_{\bmi\bmj} \xi\right), \qquad \xi \equiv \Omega^{-1}\left(\nabla_{\bm U} \Omega \chi + 3s\right)\label{xiDef}, \qquad \Sigma_{\bmd} \equiv \nabla_{\bmd}\Omega.
\end{eqnarray}
\end{subequations}
It has also been made use of the following relations,
\begin{subequations}
\begin{eqnarray}
&& \chi_{\bmi\bmj} = \Omega \varsigma_{\bmi\bmj} + \frac{1}{3}\left(\nabla_{\bm U}\Omega\right)^{-1}\left(\Omega \xi - 3s\right)g_{\bmi\bmj},\\
&&\varsigma_{\bmi\bmj} = -\left(\nabla_{\bm U}\Omega\right)^{-1}\left(D_{\bmi}f_{\bmj} - f_{\bmi}f_{\bmj} - L_{\bmi\bmj} - \frac{1}{3}\left(D_{\bmk}f^{\bmk} - f_{\bmk}f^{\bmk} - L_{\bmk}{}^{\bmk}g_{\bmi\bmj}\right)\right),\\
&&\xi = -D_{\bmi}f^{\bmi} + f_{\bmi}f^{\bmi} + L_{\bmi}{}^{\bmi} - \frac{3}{8}\Omega\rho,\label{xiIdentity}\\
&& -L_{\bm0 \bm0} + g^{\bmi\bmj}L_{\bmi\bmj} = L_{\bmb}{}^{\bmb} = \frac{1}{6}R.\label{TraceOfL}
\end{eqnarray}
\end{subequations}

\subsection{Relation to the physical field equations}

The equations \eqref{ConformalEq1}-\eqref{ConformalEq5} and \eqref{ConformalMatterEq1}-\eqref{ConformalMatterEq2} are evolution equations to the reduced system. This is, however, only a subset of the full Einstein-frame-equations. The remaining equations are constraints. It is therefore necessary to show that these constraints propagate, which indeed has been done in \cite{Friedrich2017}. We thus have the following theorem adapted from Friedrich:

\begin{theorem}
A solution 
\[
\mathbf{u} = \left(e_{\bmi}{}^{\mu}, \Gamma_{\bmk}{}^{\bmi}{}_{\bmj}, f_{\bmd}, \varsigma_{\bmi\bmj},\xi,\Omega,\Sigma_{\bmd}, s, L_{\bm0\bmi}, L_{\bmi\bmj}, \rho, \omega_{\bmi\bmj}, \omega^{\ast}{}_{\bmi\bmj} \right)\]
to the symmetric hyperbolic system \eqref{ReducedEvolEq1}-\eqref{ReducedEvolEqLast} satisfying the constraint equations associated to the conformal equations \eqref{ConformalEq1}-\eqref{ConformalEq5} and \eqref{ConformalMatterEq1}-\eqref{ConformalMatterEq2} on an initial hypersurface  implies a solution to the Einstein-$\lambda$-dust system \eqref{PhysicalEFE}-\eqref{PhysicalMatterEq2} whenever $\Omega\neq 0$.
\end{theorem}

\noindent In the following we will consider two different types of initial hypersurfaces for the conformal evolution equations \eqref{ReducedEvolEq1}-\eqref{ReducedEvolEqLast}:

\begin{itemize}
    \item[(i)] the conformal boundary $\mathcal{I}^+$;
    
    \item[(ii)] standard Cauchy hypersurface $\mathcal{S}_\star$. 
\end{itemize}

\noindent As it will be seen in more detail in the sequel, initial data for the conformal evolution equations on a standard hypersurface can be obtained from the solution of the Hamiltonian and momentum constraints implied by the Einstein-$\lambda$-dust system \eqref{PhysicalEFE}-\eqref{PhysicalMatterEq2}. 

\bigskip
 \bigskip
 \bigskip
 \bigskip

\section{Backward evolution of self-gravitating dust balls}
\label{secPastStability}
The purpose of this section is to study the (backward) evolution of asymptotic initial data for the conformal Einstein-$\lambda$-dust system which describes a collection of self-gravitating dust balls.\\

\noindent The following result is obtained by assuming certain gauge choices on a hyper surface $\mathcal{S}$ which later is interpreted to be the conformal boundary $\mathcal{I}^+$. These are,
\[
s = 0, \qquad \xi_{\bmi\bmj} = 0, \qquad \nabla_{\bma}\nabla_{\bmb}\Omega = 0, \qquad L_{\bm0 \bmi} = L_{\bmi\bm0} = 0.
\]
Note, however, that these choices may not be satisfied if one evolves a solution from $\mathcal{M}$ to $S$. In that case, the reader is referred to the original article \cite{Friedrich2017} for details.
\subsection{Asymptotic initial data for self-gravitating dust balls}
 All throughout it is assumed that the initial hyper surface $\mathcal{I}^+$ representing the conformal boundary is a compact 3-manifold. For (asymptotic) initial data prescribed on an hyper surface corresponding to the conformal boundary the following holds:

\begin{lemma}
\label{lemma1}
Any smooth initial data set for the conformal evolution equations \eqref{ReducedEvolEq1}-\eqref{ReducedEvolEqLast} is uniquely determined on $\mathcal{I}^{+}$ by a Riemannian metric $h_{ij}$, the density $\rho \geq 0$, the acceleration $f_{i}$ and symmetric, $\bmh$-tracefree tensor field $\omega_{ij}$, which are arbitrary up to the relation
\begin{equation}
D^{\bmi}\omega_{\bmi\bmj} = \frac{1}{3}D_{\bmj}\rho - \rho f_{\bmj},\label{DivergenceOmega}
\end{equation}
on $\mathcal{I}^+$, and where $D$ denotes the Levi-Civita operator defined by $h_{ij}$.
\end{lemma}
\noindent Observe that $\rho$ is allowed to be zero. This suggest to consider a density profile which represents patches of dust in an otherwise empty space. In the case of a strictly positive density function, the data $\rho, \bm\omega$ and $\bmh$ can be prescribed freely, and equation \eqref{DivergenceOmega} is read as a defining equation for the acceleration $\bm f$, unless one has further conditions on $\bm f$ such as hyper surface orthogonality etc., in which case the equation must be treated as a differential equation.
\medskip

\noindent In the following it will be shown how the above result can be used to construct asymptotic initial data representing a collection of dust balls. \medskip

\noindent The starting observation of our analysis is the fact that equation \eqref{DivergenceOmega} in Lemma \ref{lemma1} is an underdetermined condition (3 equations) for the 5 independent components of the tracefree tensor $\omega_{\bmi\bmj}$. Nevertheless, this type of divergence equations are well understood in the context of the analysis of the momentum constraint ---see e.g. \cite{Choquet-Bruhat2009}. 

\noindent In the following it will be convenient to define
\begin{equation}
\label{DefinitionNewVariable}
    \varpi_{\bmi\bmj} \equiv D_{\bmi}s_{\bmj} + D_{\bmj}s_{\bmi} - \frac{2}{3}h_{\bmi\bmj}D^{\bmk}s_{\bmk} + \Psi'_{\bmi\bmj},
\end{equation}
where $\Psi'_{\bmi\bmj}$ is a symmetric $\bmh$-tracefree tensor field which may be freely specified and $s_{\bmi}$ is an arbitrary covector field. A direct computation shows that $\varpi_{\bmi\bmj}$ is a solution of \eqref{DivergenceOmega} if $s_{\bmi}$ satisfies
\begin{equation}
    \label{ElipticEquation}
    \Delta_{\bmh}s_{\bmj} + D^{\bmi}D_{\bmj}s_{\bmi} - \frac{2}{3}D_{\bmj}D_{\bmk}s^{\bmk} = k_{\bmj} - D^{\bmi}\Psi'_{\bmi\bmj},
\end{equation}
where we have defined
\[
k_{\bmi} \equiv \frac{1}{3}D_{\bmj}\rho - \rho f_{\bmj}.
\]
Equation \eqref{ElipticEquation} is of elliptic type ---in particular, it provides 3 equations for the 3 components of $s_i$. It is convenient to reformulate the above equations by defining the operators
\[
\bm\delta \left(\bm\omega\right)_{\bmj} \equiv D^{\bmi}\omega_{\bmi\bmj}, \qquad \bmL\left(\bms\right)_{\bmi\bmj} \equiv D_{\bmi}s_{\bmj} + D_{\bmj}s_{\bmi} - \frac{2}{3}h_{\bmi\bmj}D^{\bmk}s_{\bmk}, \qquad \mathcal{L}\left(\bms\right)_{\bmj} \equiv \Delta_{\bmh}s_{\bmj} + D^{\bmi}D_{\bmj}s_{\bmi} - \frac{2}{3}D_{\bmj}D_{\bmk}s^{\bmk}. 
\]
We shall refer to these throughout as the \textit{divergence operator}, the \textit{conformal Killing operator} and the \textit{vector Laplacian operator}, respectively. It is readily seen that the vector Laplacian operator is a result of the composition of the divergence and conformal Killing operator ---i.e. we have
\begin{equation}
\mathcal{L}\left(\bms\right) = \left(\bm\delta \circ \bm L\right) \left(\bm s\right).\label{EllipticOperatorComposition}
\end{equation}
In terms of the above definitions, equations \eqref{DivergenceOmega} and \eqref{ElipticEquation} take the simple form
\begin{subequations}
\begin{eqnarray}
&& \bm\delta \left(\bm\omega\right)_{\bmj} = k_{\bmj},\label{DivergenceOmegaOperatorEq}\label{DivergenceConstraint}\\
&& \mathcal{L}\left(\bms\right)_{\bmj}  = k_{\bmj} - D^{\bmi}\Psi'_{\bmi\bmj}.\label{ElipticOperatorEq}
\end{eqnarray}
\end{subequations}
A solution $\bm s$ of equation \eqref{ElipticOperatorEq} solves equation \eqref{DivergenceConstraint} if the symmetric tracefree tensor $\bm\omega$ is of the form given by \eqref{DefinitionNewVariable}. To solve the elliptic equation for the covector $\bm s$ we make use of the following \cite{Eva98}: 

\begin{fact}[Fredholm alternative]
\label{FredholmTheorem}
 Given any $\bmu$ and $\bmv$ $\in L^2$, then there exists a solution $\bmu$ of the elliptic equation \[
 \mathcal{\bmL}\left(\bm u\right)= \bm F
\]
if there exists a $\bmv$ which solves $ \mathcal{\bm L}^{\ast}\left(\bm v\right) = 0$ and satisfy the $L^2$-inner product
\begin{equation}
\langle{\bm v,\bm F\rangle} = \int_{\mathcal{S}} h^{ij} v_i F_j \mathrm{d}\mu = 0.
\label{FredholmObstruction}
\end{equation}
\end{fact}

\noindent The operators $\bm\delta$ and $\bmL$ can be regarded as formal adjoints of each other under the standard $L^2$-inner product over a compact 3-manifold $\mathcal{S}$. It then follows that their composition, the operator $\mathcal{L}$, is self-adjoint ---that is, 
 \[
 \langle{\bmu, \mathcal{L}\left(\bms\right)\rangle}=\langle{\mathcal{L}\left(\bmu\right), \bms\rangle}.
 \]

\medskip
\noindent In order to make use of the Fredholm alternative to establish the existence of solutions to equation \eqref{ElipticOperatorEq} it is necessary to identify the Kernel of the operator $\mathcal{L}$. For this, it is observed that
\begin{eqnarray*}
&& 0=\langle \bmv, \mathcal{L}(\bmv) \rangle = \langle \bmv, ({\bm \delta}\circ \bmL)(\bmv) \rangle \\
&& \phantom{0}= \langle {\bm \delta}^*(\bmv), \bmL (\bmv) \rangle = \langle \bmL(\bmv), \bmL(\bmv) \rangle.
\end{eqnarray*}
Consequently, any element of the Kernel of $\mathcal{L}$ satisfies the equation $\bmL(\bmv)=0$ ---that is, the Kernel consists of conformal Killing vectors. Thus, if the pair $(\mathcal{S},\bmh)$ does not have conformal Killing vectors (this is the generic situation) then there are no obstructions to the existence of solutions to equation \eqref{ElipticOperatorEq}. On the other hand, if conformal Killing vectors are present then the Kernel orthogonality condition in the Fredholm alternative, equation \eqref{FredholmObstruction}, has to be satisfied. 

\noindent The discussion of the previous paragraph is summarised in the following result where all the relevant fields are assumed to be suitably smooth:
\begin{lemma}
\label{Lemma:ExistenceAsymptoticData}
Let $\mathcal{S}$ denote a compact 3-dimensional manifold. Given a (Riemannian) metric $h_{ij}$, a $\bmh$-tracefree tensor $\Psi'_{ij}$, a covector $f_i$ and a scalar $\rho$ over $\mathcal{S}$ then one of the following holds:
\begin{itemize}
    \item[(i)] if $(\mathcal{S},h_{ij})$ admits no conformal Killing vectors then the tracefree tensor $\varpi_{ij}$ given by equation \eqref{DefinitionNewVariable} gives a solution to the asymptotic constraint \eqref{DivergenceOmega};
    
    \item[(ii)] if $(\mathcal{S},h_{ij})$ admits conformal Killing vectors then $\varpi_{ij}$ given by equation \eqref{DefinitionNewVariable} gives a solution to the asymptotic constraint \eqref{DivergenceOmega} if and only if
    \[
    \int_{\mathcal{S}}v^{\bmi}\left(k_{\bmi} - D^{\bmj}\Psi'_{\bmi\bmj}\right)d\mu = 0
    \]
    for any conformal Killing vector $v^i$.
\end{itemize}

\end{lemma}

\begin{remark}
{\em In the present context, the simplest example of a pair $(\mathcal{S},h_{ij})$ with conformal Killing vectors is the 3-sphere $\mathbb{S}^3$ with the round metric. In this case one has, in fact, the maximal number of conformal Killing vectors (10) for a 3-dimensioanl manifold.}
\end{remark}

\begin{remark}
{\em The freely specifiable data given by the tracefree tensor $\Psi'_{ij}$ can be thought of as describing some gravitational wave content.}
\end{remark}

\medskip
\noindent In order to construct initial data representing a collection of balls of dust, let $\Sigma_i$, $i=1,\ldots,n$ denote $n$ compact open subsets on $\mathcal{S}$ and consider a smooth non-negative scalar field $\rho$ over $\mathcal{S}$ with support on the union of the sets $\Sigma_i$. That is, we require that 
\begin{equation}
\label{BlubOfDustDef}
\begin{cases} 
      \rho > 0, & \rho \in \bigcup\limits_{i=1}^{n} \Sigma_{i}, \\
     \rho = 0, & \rho \in \mathcal{I}^+/\bigcup\limits_{i=1}^{n} \Sigma_{i}.
   \end{cases}
\end{equation}
Lemma \ref{Lemma:ExistenceAsymptoticData} gives the conditions for the existence of solution to the asymptotic constraint \eqref{DivergenceOmega} for this type of density profile $\rho$ and a given choice of metric $h_{ij}$ and fields $\Psi'_{ij}$ and $f_i$.

\begin{remark}
{\em Let $\rho$ be a smooth non-negative scalar field given by $\eqref{BlubOfDustDef}$. Furthermore, let $\bm\Psi'$ be a symmetric, $\bmh$-tracefree spatial tensor field, $\bm h$ the projector metric and $\bm f$ a one form, then we say that
$\left(\rho, \bmh, \bm\Psi', \bmf \right)$ is \textbf{n-body dust asymptotic data} if either condition (i) or (ii) of Lemma \ref{Lemma:ExistenceAsymptoticData} holds on $\mathcal{I}^+$.}
\end{remark}

\subsection{Evolution of the asymptotic data}
The asymptotic data constructed in the previous subsection  can be readily combined with the conformal evolution equations of Section \ref{Section:EvolutionEquations} to obtain the asymptotic region of a spacetime with positive Cosmological constant containing a collection of $n$ balls of dust. The key observation here is that as we are working in the conformal picture, any interval of time of the conformal boundary represents an infinite time domain from the physical perspective. The existence result can be stared as follows:

\begin{theorem}
\label{theorem_dataOnScri}
Given a choice of asymptotic data  representing a collection of $n$ dust balls, there exists a time $\tau>0$ such that the conformal Einstein-$\lambda$-dust equations have a unique smooth solution on the slab $[0,\tau)\times\mathcal{S}$ associated to this data. This solution implies, in turn, a solution to the (physical) Einstein-$\lambda$-dust system on $(0,\tau)\times\mathcal{S}$ for which the hypersurface $\{0\}\times \mathcal{S}$ corresponds to the conformal boundary $\mathcal{I}^+$. 
\end{theorem}



\begin{remark}
{\em By restricting the existence time further, if necessary, it is possible to ensure that the congruence of conformal geodesics on which our gauge is based remains non-intersecting for the interval $[0,\tau]$. This, in turn, ensures that dust balls in the initial asymptotic configuration do not intersect each other in the past. }
\end{remark}

\noindent In terms of physics, theorem \ref{theorem_dataOnScri}, suggest the following. If, in the infinite far future of an expanding universe, one is given a matter distribution representing patches of dust balls, then one can evolve this system backward in time for as long as one wish, and still have that the patches of dust remain non-interacting.

\section{Forward evolution of dust balls}
\label{sec:futureStability}

In this section we consider the more physically realistic setting of the evolution of dust balls from a standard Cauchy hypersurface in a spacetime with positive Cosmological constant. Our strategy is to consider this setting as a perturbation of the de Sitter spacetime in order to make a statement of the future global existence of the dust balls. As in the case of the backwards evolution we start by constructing suitable initial data.

\subsection{Standard Cauchy initial data for self-gravitating dust balls}

Let $\tau \in \mathcal{\tilde{M}}$ be a positive function such that for $t \in \mathds{R}$, $\tau(p) = t$ gives the level surfaces $\tilde{\mathcal{S}}_t$. We denote by $\tilde{\mathcal{S}}_\star  \subset \mathcal{\tilde{M}}$ the hypersurface which coincides with the level surface $\tau(p) = 0$, and interpret this as an initial hypersurface at some fiduciary time.
\emph{Einstein constraints} on $\tilde{\mathcal{S}}_\star$ are given by
\begin{subequations}
\begin{eqnarray}
&& r[\tilde{\bmh}] + \tilde{K}^2 - \tilde{K}_{ij}\tilde{K}^{ij} = 2\left( \tilde{\rho} - \lambda\right),\label{EinsteinConstraintEq1}\\
&& \tilde{D}^i\tilde{K}_{ij} - \tilde{D}_j \tilde{K} = - \tilde{j}_j,\label{EinsteinConstraintEq2}
\end{eqnarray}
\end{subequations}
where $\tilde{\bmh}$ and $\tilde{K}_{ij}$ denote, respectively, the intrinsic metric and extrinsic curvature of $\tilde{\mathcal{S}}_\star$, $\tilde{D}$ is the Levi-Civita connection of  metric $\tilde{\bmh}$ and $r[\tilde{\bmh}]$ its Ricci scalar. Moreover, $\tilde{\rho}$ and $\tilde{\bmj}$ denote, respectively, the \emph{energy-density} and \emph{flux current} of the matter content.

\medskip
The constraints \eqref{EinsteinConstraintEq1} and \eqref{EinsteinConstraintEq2} will be solved using the \emph{conformal method of Licnerowicz-York} ---see e.g. \cite{Choquet-Bruhat2009}. Following the discussion in \cite{Kroon2016}, Chapter 11, let $h_{ij}=\Omega^2 \tilde{h}_{ij}$. Implementing this rescaling in equations \eqref{EinsteinConstraintEq1}-\eqref{EinsteinConstraintEq2} leads to 
\begin{subequations}
\begin{eqnarray}
&&2\Omega D_iD^i\Omega - 3D_i\Omega D^i\Omega + \frac{1}{2}\Omega^2 r + 3\Sigma^2+\frac{1}{2}\Omega^2\left(K^2 - K_{ij}K^{ij}\right) -2 \Omega \Sigma K = \Omega^4\rho- \lambda , \label{ConfEinsteinConstraintEq1}\\
&&\Omega^3 D^i K_{ij} - 2K^i{}_j D_i\Omega - \Omega D_k K + 2D_k\Sigma = \Omega^3 j_k,\label{ConfEinsteinConstraintEq2},
\end{eqnarray}
\end{subequations}
where
\[
\rho \equiv \Omega^{-4}\tilde{\rho}, \qquad j_k \equiv \Omega^{-3}\tilde{j}_k, \qquad \Sigma \equiv \nu^i D_i\Xi,
\]
and $\Omega = \Xi|_{\tilde{\mathcal{S}}_\star}$. Now, by setting $\Omega = \theta^{-2}$, equation \eqref{ConfEinsteinConstraintEq1} leads to the manifestly elliptic equation 
\begin{equation}
\label{ConfEllipticEq1}
L_{\bmh}\theta = \frac{1}{8}\theta\left(K^2 - K_{\bmi\bmj}K^{\bmi\bmj}\right) - \frac{1}{2}\theta^2 \Sigma K - \frac{1}{4}\theta^5\left(\theta^{-8}\rho-\lambda \right),
\end{equation}
where we have defined the \textit{Yamabe operator}
\[
L_{\bmh}\theta \equiv D_i D^i\theta - \frac{1}{8} r[\bmh] \theta.
\]
Now, defining 
\[
\psi_{ij} \equiv \theta^{4}K_{\{ij\}}, \qquad K_{\{ij\}} = K_{j} - \frac{1}{3}K h_{ij},
\]
it follows that equation \eqref{ConfEinsteinConstraintEq2} leads to the equation
\begin{equation}
\label{DivergencePsi}
D^{i}\psi_{ij} = \frac{2}{3}\theta^6 D_{j}\tilde{K} - 2\theta^{6}D_{j}\Sigma + j_{j}.
\end{equation}

\begin{remark}
{\em We will consider equations \eqref{ConfEllipticEq1} and \eqref{DivergencePsi} in the particular case that 
\[
K = \Sigma = 0.
\]
It can be readily verified that the above conditions imply that $\tilde{\mathcal{S}}_\star$ is a maximal hypersurface.}
\end{remark}

In order to put equation \eqref{DivergencePsi} in an elliptic form, we make use of the \textit{York splitting} ---i.e. given an arbitrary covector field $X_i$, we consider solutions $\psi_{ij}$ of the form
\begin{equation}
\label{YorkSplit}
\psi_{ij} = \left(\mathcal{L}_{\bmh}X\right)_{ij} + \psi'_{ij},
\end{equation}
where $\psi'_{ij}$ is a freely specifiable symmetric and tracefree tensor field, and $\mathcal{L}_{\bmh}{\bm X}$ is the \textit{conformal Killing operator} defined by
\[
\left(\mathcal{L}_{\bmh}X\right)_{ij} \equiv D_{i}X_{j} + D_{j}X_{i}- \frac{2}{3}h_{ij}D_{k}X^{k}. 
\]
For simplicity, we set $\psi'_{ij}=0$, so that substituting \eqref{YorkSplit} into equation \eqref{DivergencePsi}, we obtain the elliptic equation
\begin{equation}
\label{ConfEllipticEq2}
    D^{i}\left(\mathcal{L}_{\bmh}X\right)_{ij} = j_j.
\end{equation}

We thus seek to show that there exist a solution to the elliptic equations \eqref{ConfEllipticEq1} and \eqref{ConfEllipticEq2} which represents initial data for a de Sitter-like spacetime with an energy density function given by \eqref{BlubOfDustDef} ---so that it can be regarded as describing a collection of dust balls.

\medskip
With regards to the solution to equation \eqref{ConfEllipticEq2} we adapt the following result from  \cite{Choquet-Bruhat2009}, Chapter VII, Section 6: 

\begin{proposition}
\label{Proposition:York}
 Let $\bmh \in H^2(\tilde{\mathcal{S}}_\star)$ and $\bm\xi$ be, respectively, a Riemannian metric and a conformal Killing vector over $\tilde{\mathcal{S}}_\star$. Then equation \eqref{ConfEllipticEq2} has a solution $\bm X \in H^2(\tilde{\mathcal{S}}_\star)$ if $\bm j \in L^2(\tilde{\mathcal{S}}_\star)$ and  
\[
\int_{\tilde{\mathcal{S}}_\star}
\bmh^\sharp( \bmj,\bm\xi)
 \bm\epsilon_{\bmh} = 0.
\]
The solution is determined up to the addition of a conformal Killing vector. Furthermore, the solution is unique if one imposes
\[
\int_{\tilde{\mathcal{S}}_\star} \langle {\bm X},\bm\xi\rangle \bm\epsilon_{\bmh} = 0.
\]
In that case there exists a positive constant $C$ such that 
\[
\norm{\bm X}_{L^2}^2 \leq C \norm{\bmj}_{L^2}^2.
\]
\end{proposition}

Now, setting $K = \Sigma = 0$ and the using tracefree tensor $\psi_{ij}$ defined in equation \eqref{YorkSplit}, the Licnerowicz equation \eqref{ConfEllipticEq1} can be written as
\begin{equation}
\label{LichnerowiczEq}
D_{i}D^{i}\theta - a\theta + b\theta^{-7}  + c\theta^5 =0,
\end{equation}
where,
\[
a \equiv \frac{1}{8}\bm r[\bmh], \qquad b \equiv \frac{1}{8} \psi_{ij}\psi^{ij}, \qquad c \equiv \frac{1}{4}\left( \tilde{\rho}- \lambda\right), \quad \tilde{\rho} = \Omega^4 \rho = \theta^{-8}\rho.
\]
Following the theory developed in \cite{Choquet-Bruhat2009} Chapter VII, Sections 5, 6 and 7  (see also \cite{Choquet-Bruhat1976}) the above equation has a unique solution $\theta > 0$ if  $b \geq 0$ and $c < 0$. Since one readily has that $\psi_{ij}\psi^{ij} > 0$, the only condition to be imposed is
\[
\tilde{\rho} < \lambda.
\]
Thus, one has the following
\begin{proposition}
\label{lemmaLichnerowiczSoln}
For $r[\bmh]>0$, $\psi_{ij}\psi^{ij}>0$ and $\lambda >0$, the condition $\tilde{\rho} < \lambda$ is a sufficient condition for the existence of a unique solution $\theta$ to the Lichnerowicz equation \eqref{LichnerowiczEq}.
\end{proposition}

Together, Propositions \ref{Proposition:York} and \ref{lemmaLichnerowiczSoln} ensure the existence of a large class of solutions to the Einstein constraint equations representing an arbitrary configuration of dust balls at a some fiduciary time. For this, as in the asymptotic problem, one chooses the density $\tilde{\rho}$ as in eqution \eqref{BlubOfDustDef} ---the method for the construction of solutions to the Einstein constraints described above works irrespectively from the fact that the density is only non-zero on a finite number of subsets of $\tilde{\mathcal{S}}_\star$. If, in addition, one chooses the metric $\bmh$ as a constant multiple of the round metric on $\mathbb{S}^3$ ---as in the case of the de Sitter spacetime--- one can then regard the dust balls as matter-sourced perturbation of the de Sitter spacetime. The size of $\tilde{\rho}$ as described in terms of Sobolev norms controls the closeness of $\theta$ to the value $1$ (the de Sitter value). This observation is of importance in the discussion of the stability of solutions to the evolution problem.

\begin{remark}
\label{Remark:ParticulariseData}
{\em For the purpose of simplicity of presentation of the subsequent discussion it is convenient to consider a setting in which the initial current vector $\bmj$ vanishes. This choice of free data is consistent with the 4-velocity $\bmu$ being orthogonal to the initial hypersurface $\tilde{\mathcal{S}}_\star$. This choice is made throughout the whole hypersurface regardless of whether the density vanishes or not in a given region. For this choice, if the density vanishes all over the initial hypersurface, then one obtains \emph{trivial data} corresponding to the de Sitter spacetime. }
\end{remark}

\medskip
Following the discussion in \cite{Kroon2016} Chapter 11, from a solution to the Einstein constraint equations it is possible to obtain a solution to the conformal Einstein field equations by algebraic manipulations and differentiation. The deviation of this data from (vacuum) data for the de Sitter spacetime is controlled by the size of the current $\bmj$ and the density $\tilde{\rho}$. 

\subsection{Long time evolution}

In this section we discuss the evolution of the initial data given by Propositions \ref{Proposition:York} and \ref{lemmaLichnerowiczSoln}. In particular, we discuss how the ideas used in the stability of the de Sitter spacetime \cite{friedrich1986a} (see also \cite{Kroon2016}, Chapter 15) can be used to obtain a future global existence statement for the dust balls if the initial density is sufficiently small.

\medskip
In the following let $\mathbf{u}$ denote a solution to the conformal evolution equations discussed in Section \ref{Section:EvolutionEquations}. Moreover, let $\mathring{\mathbf{u}}$ denote the solution to these evolution equations with $\rho=0$ (i.e. vanishing density) and the 4-velocity $u^a$ chosen so that it is tangent to timelike geodesics in the physical spacetime ---see Remark \ref{Remark:ParticulariseData}.  Denote by $\mathbf{u}_\star$ and $\mathring{\mathbf{u}}_\star$ the associated initial data on some fiduciary initial hypersurface $\mathcal{S}_\star$. The solution $\mathbf{u}$ provides a conformal representation of the de Sitter spacetime which is smooth up to and beyond the conformal boundary $\mathcal{I}^+$. In particular, it has vanishing rescaled Weyl tensor. For concreteness assume that the conformal boundary for this (background) solution is given by the condition $\tau=\tau_\infty$, for $\tau_\infty$ some constant. To this background solution one can readily apply the standard theory of stability for symmetric hyperbolic equations ---see \cite{Kat75}; also \cite{Kroon2016}--- to ensure the existence of nearby solutions (in the sense of Sobolev spaces) to the evolution equations with a similar existence time. Accordingly, these solutions extend up to and beyond the conformal boundary. This amounts to a future global existence result. More precisely, one has the following:

\begin{theorem}
\label{FutureGlobalExistenceDustBalls}
Let $\mathbf{u}_\star$ denote smooth initial data for the conformal-$\lambda$-dust evolution equations on a compact manifold $\mathcal{S}_\star$ describing a configuration of dust balls as given by Propositions \ref{Proposition:York} and \ref{lemmaLichnerowiczSoln}. There exists $\varepsilon>0$ such that for any initial data $\mathbf{u}_\star$ such that
\[
\norm{\mathbf{u}_\star -\mathring{\mathbf{u}}_\star}_m <\varepsilon, \qquad m\geq 5,
\]
 there exists a smooth solution $\mathbf{u}$ to the conformal evolution equations over the domain
\[
\mathcal{M}\equiv [\tau_\star,\tau_\infty] \times \mathcal{S},
\]
$\mathcal{S}\approx \mathcal{S}_\star$. Moreover, given a sequence of initial data $\mathbf{u}^{(n)}_\star$, as above, such that
\[
\norm{\mathbf{u}_\star -\mathring{\mathbf{u}}_\star}_{m}\rightarrow 0, \qquad \mbox{as} \quad n\rightarrow \infty,
\]
one has that the corresponding solutions satisfy 
\[
\norm{\mathbf{u}(\tau,\cdot) -\mathring{\mathbf{u}}(\tau,\cdot)}_{m}\rightarrow 0, \qquad \mbox{as} \quad n\rightarrow \infty.
\]
The solution $\mathbf{u}$ implies, in turn, a future geodesically complete solution to the (physical) Einstein-$\lambda$-dust system for which $\mathcal{I}^+$ corresponds to future (timelike) infinity.
\end{theorem}

\begin{proof}
The proof of this result follows the same structure of that of the stability of the de Sitter spacetime \cite{friedrich1986a,Friedrich1986b} ---see also \cite{Kroon2016}, Chapter 15. Here we provide a brief outline of the main ideas. As already mentioned, the evolution equations \eqref{ReducedEvolEq1}-\eqref{ReducedEvolEqLast} implies a symmetric hyperbolic evolution system for the components of the vector unknown $\mathbf{u}$. Now writing $\mathbf{u}=\mathring{\mathbf{u}}+\breve{\mathbf{u}}$ where $\mathring{\mathbf{u}}$ denotes the \emph{background} de Sitter solution, it follows that the \emph{perturbation} $\breve{\mathbf{u}}$ also satisfies a symmetric hyperbolic evolution system. Existence of solutions for this system follows from the theory developed in \cite{Kat75}. Moreover, as the perturbed initial data $\breve{\mathbf{u}}_\star$ is small (in the sense of Sobolev spaces), it follows then from Cauchy stability that its existence interval includes the time $\tau_\infty$ ---so that the development includes the conformal factor. Finally, a \emph{propagation of the constraints} argument ensures the the solution to the reduced evolution system implies a solution to the physical Einstein-$\lambda$-dust system.
\end{proof}

\begin{remark}
{\em From the discussion leading to Propositions \ref{Proposition:York} and \ref{lemmaLichnerowiczSoln}, it follows that the size (in the Sobolev norm) of the initial data $\mathbf{u}_\star$ is controlled by the initial value of the density over $\mathcal{S}_\star$. In particular, if $\rho_\star=0$ then $\mathbf{u}_\star=\mathring{\mathbf{u}}_\star$. Accordingly, Theorem \ref{FutureGlobalExistenceDustBalls} states that the initial configuration of dust balls will exist globally into the future if the density is sufficiently small ---that is, if the dust making up the balls is sufficiently diluted.
 }
\end{remark}

\begin{remark}
{\em The spacetimes arising from Theorem \ref{FutureGlobalExistenceDustBalls} can be readily shown to be geodesically complete. The simplest manner of doing this is to make use of the theory developed in \cite{ChoCot02}. The required estimates needed to establish geodesic completeness follow from the closeness (in the sense of Sobolev spaces) of the solution provided by Theorem \ref{FutureGlobalExistenceDustBalls} and the background exact de Sitter solution. In the present case it is possible to show even more: as the background 4-velocity $\mathring{u}^a$ is chosen to be tangent to a congruence of non-intersecting conformal geodesics, it follows that if the perturbed solutions given by Theorem \ref{FutureGlobalExistenceDustBalls} are the flow lines of $u^a$, then they are also non-intersecting. This observation shows, in addition, that the various members of an arbitrary configuration of dust balls never intersect in the future.}
\end{remark}



\label{Section:FinalRemarks}
The purpose of this article is the development of a model of self-gravitating bodies in General Relativity for which it is possible to make statements of long-term existence. As mentioned in the introduction, the well-posedness and local existence in time of self-gravitating balls of dust has been given in \cite{Choquet-Bruhat2006}. These self-gravitating bodies possess a smooth boundary (in the sense that the density is assumed to go to zero smoothly). This observation, combined with an evolution law for the 4-velocity which is well defined even in the regions where the density vanishes allows to obtain a suitable evolution system for which existence theory is available. The analysis of the Einstein-$\lambda$-dust system in \cite{Friedrich2017} provides a conformal analogue to this system and thus, it allows to implement an argument establishing long-term existence of dust ball configurations. The physical mechanism making it possible to run this argument is the \emph{acceleration} provided by the Cosmological constant $\lambda$. 
\medskip
It should be mentioned that an extension of this result to an asymptotically flat setting (where $\lambda=0$) is made much more challenging by the fact that in this scenario, and following a conformal point of view, timelike geodesics converge at future timelike infinity $i^+$. Accordingly, any attempt to analyse the long-term existence of matter configurations is tied to the development of a suitable description of this asymptotic point. 

\bibliography{References/references}

\begin{thebibliography}{10}

\bibitem{Choquet-Bruhat1976}
Y.~Choquet-Bruhat.
\newblock {\em The Problem of Constraints in General Relativity: Solution of
  the Lichnerowicz Equation}, pages 225--235.
\newblock Springer Netherlands, Dordrecht, 1976.

\bibitem{Choquet-Bruhat2009}
Y.~Choquet-Bruhat.
\newblock {\em {General Relativity and the Einstein Equations}}.
\newblock Oxford Mathematical Monographs. Oxford University Press, United
  Kingdom, 2009.

\bibitem{ChoCot02}
Y.~Choquet-Bruhat and S.~Cotsakis.
\newblock Global hyperbolicity and completeness.
\newblock {\em J. Geom. Phys.}, 43:345, 2002.

\bibitem{Choquet-Bruhat2006}
Y.~Choquet-Bruhat and H.~Friedrich.
\newblock Motion of isolated bodies.
\newblock {\em Classical and Quantum Gravity}, 23(20):5941--5949, sep 2006.

\bibitem{Eva98}
L.~Evans.
\newblock {\em Partial differential equations}.
\newblock American Mathematical Society, 1998.

\bibitem{Frauendiener2004}
J.~Frauendiener.
\newblock Conformal infinity.
\newblock {\em Living Reviews in Relativity}, 7(1):1, 2004.

\bibitem{Friedrich1986b}
H.~Friedrich.
\newblock Existence and structure of past asymptotically simple solutions of
  einstein's field equations with positive cosmological constant.
\newblock {\em Journal of Geometry and Physics}, 3(1):101 -- 117, 1986.

\bibitem{friedrich1986a}
H.~Friedrich.
\newblock On the existence of $n$-geodesically complete or future complete
  solutions of einstein's field equations with smooth asymptotic structure.
\newblock {\em Comm. Math. Phys.}, 107(4):587--609, 1986.

\bibitem{Friedrich1991}
H.~Friedrich.
\newblock On the global existence and the asymptotic behavior of solutions to
  the einstein-maxwell-yang-mills equations.
\newblock {\em J. Differential Geom.}, 34(2):275--345, 1991.

\bibitem{Friedrich2014}
H.~{Friedrich}.
\newblock {Geometric Asymptotics and Beyond}.
\newblock {\em arXiv e-prints}, page arXiv:1411.3854, Nov. 2014.

\bibitem{Friedrich2017}
H.~{Friedrich}.
\newblock {Sharp Asymptotics for Einstein-\{{\ensuremath{\lambda}}\}-Dust
  Flows}.
\newblock {\em Communications in Mathematical Physics}, 350(2):803--844, Mar.
  2017.

\bibitem{H.Friedrich1981}
H.Friedrich.
\newblock On the regular and the asymptotic characteristic initial value
  problem for einstein’s vacuum field equations.
\newblock {\em Proc. R. Soc. Lond. A375169–184}, 1981.

\bibitem{Kat75}
T.~Kato.
\newblock The {Cauchy} problem for quasi-linear symmetric hyperbolic systems.
\newblock {\em Arch. Rat. Mech.}, 58:181, 1975.

\bibitem{Newman1981}
E.~Newman and K.~Tod.
\newblock {\em {Asymptotically flat space-times}}.
\newblock Plenum, 1 1981.

\bibitem{Penrose1965}
R.~Penrose.
\newblock Zero rest-mass fields including gravitation: Asymptotic behaviour.
\newblock {\em Proceedings of the Royal Society of London. Series A,
  Mathematical and Physical Sciences}, 284(1397):159--203, 1965.

\bibitem{Kroon2016}
J.~A. Valiente~Kroon.
\newblock {\em Conformal Methods in General Relativity}.
\newblock Cambridge University Press, 2016.

\bibitem{Wald1984}
R.~M. Wald.
\newblock {\em {General Relativity}}.
\newblock Chicago Univ. Pr., Chicago, USA, 1984.

\end{thebibliography}
\bibliographystyle{abbrv}
 \end{document}